\documentclass[superscriptaddress,aps,prl,amssymb,amsfonts,letterpaper,twocolumn]{revtex4}
\usepackage{amsmath} \usepackage{amsthm} \usepackage{color}
\usepackage{graphics,graphicx}
\usepackage[pdftex]{hyperref}

\usepackage{verbatim}
\usepackage{amsmath}
\usepackage{latexsym}
\usepackage{yfonts}
\usepackage{ifthen}

\usepackage{natbib}
\usepackage{amsfonts}
\usepackage{makeidx}
\usepackage{amsmath}
\usepackage{amssymb}
\usepackage{amsthm}
\usepackage{graphicx}
\usepackage{bm}
\usepackage{bbm}

\newcommand{\be}{\begin{eqnarray} \begin{aligned}}
\newcommand{\ee}{\end{aligned} \end{eqnarray} }
\newcommand{\benn}{\begin{eqnarray*} \begin{aligned}}
\newcommand{\eenn}{\end{aligned} \end{eqnarray*} }

\newcommand{\bc}{\begin{center}}
\newcommand{\ec}{\end{center}}

\newcounter{claimCount}[section]
\newtheorem{claim}{Claim}[claimCount]

\usepackage{amsfonts}

\def\01{\{0,1\}}

\newcommand{\eps}{\varepsilon}

\newcommand{\nsw}{noisy-storage model}

\newcounter{protoCount}
\newcounter{protoList}
\newsavebox{\tmpbox}
\newlength{\protobox}
\newenvironment{protocol}[2]{
\bigskip
\addtocounter{protoCount}{1}
\noindent \begin{lrbox}{\tmpbox}
\setlength{\protobox}{\columnwidth}
\addtolength{\protobox}{-0.5cm}
\begin{minipage}[c]{\protobox}
\begin{bfseries}Protocol \theprotoCount: #1\end{bfseries}
\ifthenelse{\equal{#2}{\empty}}{}{\\ #2}
\begin{list}{\begin{bfseries}\arabic{protoList}:\end{bfseries}}
{\usecounter{protoList}}
}{
\end{list}
\end{minipage}\end{lrbox}
\fbox{\usebox{\tmpbox}}
\bigskip
}

\newcommand{\COMMENT}[1]{}

\begin{document}

%\title{Two-party quantum cryptography over long physical distances made simple}
\title{Long distance two-party quantum cryptography made simple}

\author{Iordanis Kerenidis}
\affiliation{CNRS, Laboratoire de Recherche en Informatique, Univ Paris 11, Orsay, 91405 France}
\affiliation{Centre for Quantum Technologies, National University of Singapore, 2 Science Drive 3, 117543 Singapore}
\author{Stephanie Wehner}
\affiliation{Centre for Quantum Technologies, National University of Singapore, 2 Science Drive 3, 117543 Singapore}
%\affiliation{Institute for Quantum Information, Caltech, Pasadena, CA 91125, USA}
\date{\today}
\begin{abstract}
	Any two-party cryptographic primitive
	can be implemented using quantum communication under the assumption that it is difficult to store a large number of quantum states perfectly. However, achieving reliable quantum communication over long distances remains a difficult problem. 
Here, we consider a large network of nodes with only neighboring quantum links.
	%execute the universal cryptographic primitive of oblivious transfer.
	We exploit properties of this cloud of nodes to enable any two nodes to achieve security even if
	they are not directly connected. Our results are based on techniques from classical cryptography and 
	do not resort to technologically difficult procedures like entanglement swapping. 
	More precisely, we show that oblivious transfer can be achieved in such a network if and only 
	if there exists a path in the network between the sender and the receiver along which all nodes are honest.
	%We also show that oblivious transfer between two connected nodes is a necessary prerequisite to achieve security in this setting.
	Finally, we show that useful notions of security can still be achieved when we relax the assumption of an honest path. For example, we 
	show that we can combine our protocol for oblivious transfer with computational assumptions such that we obtain security if either
	there exists an honest path, or, as a backup, at least the adversary cannot solve a computational problem.
\end{abstract}
\maketitle

Quantum communication allows us to achieve cryptographic security without relying on unproven computational assumptions.
Two nodes, Alice and Bob, can establish a secure key using quantum key distribution~\cite{bb84,e91}, and, moreover, solve any two-party 
cryptographic problem even if they do not trust each other in the \nsw~\cite{prl:noisy, noisy:new, chris:new}. 
Well-known examples of such problems include secure identification~\cite{bounded:secureId}, 
as well as electronic voting and secure auctions. More generally, Alice and Bob wish to solve problems where Alice holds 
an input $x$ (eg. the amount of money she is willing to bid for an item sold by Bob) and Bob holds an input $y$ (e.g. his minimum asking price), 
and they want to obtain the value of some function $f(x,y)$ (e.g. output no if 
$x < y$, and $x$ otherwise) as depicted below. 
In this setting, there is no outside eavesdropper but Alice or Bob themselves may be dishonest. 
Security thereby means 
that Alice should not learn anything about $y$ and Bob should not learn anything about $x$, apart from what can be inferred from the value of $f(x,y)$~\cite{yao:sfe}.
\begin{center}
\includegraphics[scale=0.6]{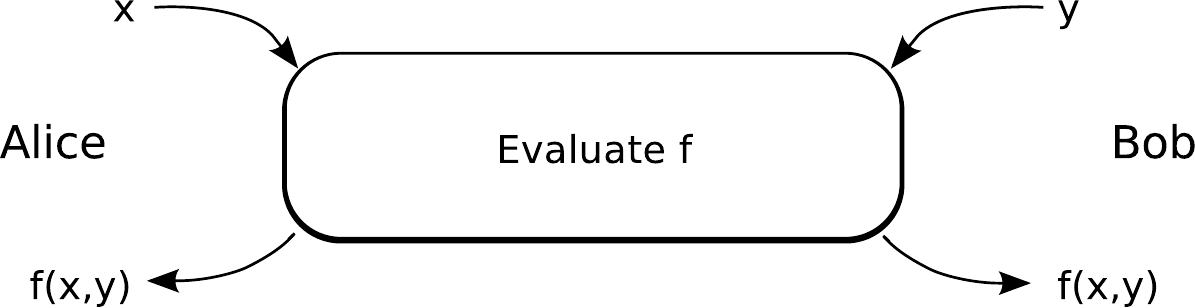}
\end{center}

Unfortunately, quantum communication over long distances poses a formidable problem. 
At present, quantum key distribution has been achieved over a distance of at most 145km in fiber~\cite{hughes:fiber} or 144km in freespace~\cite{qkd:freespace1,zeilinger:qkd}.
In addition, having a direct communication link between any two nodes that may wish to communicate is an infeasible problem even when it comes to classical
communication. Instead, we have networks of nodes, such as the present day internet, in which only some nodes are directly connected, but 
are willing to relay communication for other nodes who do not share a direct link.
Typically it is easy to connect two nodes who are physically close.
In order to achieve longer distances, many forms of quantum repeaters have been proposed in order extend the range of quantum communication to
obtain a quantum version of the internet~\cite{lloyd:internet, kimble:internet}.
Broadly speaking, quantum repeaters used in key distribution come in two variants: in the first, the nodes along
the path between Alice and Bob are trusted, and we perform quantum key distribution between each two neighbours.
This form of repeater is known as trusted relay and was for example used in the network of SECOQC~\cite{secoqc}.
The second method is to have the intermediary nodes create entanglement, allowing Alice and Bob to create entanglement between them using
the concept of entanglement swapping~\cite{entSwapping}. This is clearly more desirable than relying on trusted relays, but technologically
very difficult to achieve especially when there are many intermediary nodes. 
Many experiments have been done over the last twelve years~\cite{zeilinger:entanglement, entSwapex2, entSwapex3}, 
but still we are far from using this technology for QKD~\cite{secoqc}, and similarly for the case of two-party computation in the \nsw.
What both of these approaches have in common is that they first try to create the analog of a point-to-point link between Alice and Bob
to solve the final cryptographic task. 

Here, we take a novel approach using techniques from classical cryptography to bridge the potentially large physical distance between Alice
and Bob. 
%We thereby try to spread trust over as many nodes as possible, where we will focus exclusively on the problem of two-party secure computation. 
Concretely, we for the first time consider the case where any two nodes that are directly connected by a (quantum)
communication link can securely solve the universal cryptographic problem of oblivious transfer (OT),
which in turn enables them to solve any two-party cryptographic problem~\cite{kilian:foundingOnOT}. Implementations of such protocols ({\em link}-OT) can be found in \nsw~\cite{prl:noisy,noisy:new,chris:new}. 
Any node in the network may behave honestly, or be dishonest in the sense that it will collaborate with the dishonest Alice or Bob.
A dishonest node also has full control over the communication links attached to it, making it more powerful than for example
the eavesdropper in QKD who only has access to the communication link and not to any of the individual labs.
\begin{figure}
	\includegraphics{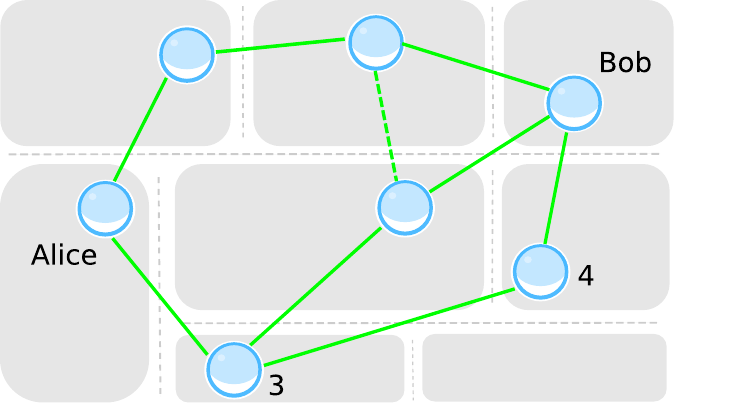}
	\caption{\scriptsize If any two nodes with a direct link can perform oblivious transfer,
	then Alice and Bob can solve any two-party cryptographic problem as long as there exists a path from Alice to Bob (e.g., $3$ and $4$) along which all intermediary nodes are honest, 
	or the cheating party cannot solve a computational problem efficiently.}
	\label{fig:nodes}
\end{figure}

{\bf Results } We first provide a simple protocol for oblivious transfer between Alice and Bob, who do not share a direct quantum link, that is secure, as long as all nodes along one of the paths from Alice to Bob
are honest ({\em path}-OT). We will refer to this path as an \emph{honest path} from Alice to Bob, which is in flavor similar to recent extensions to the idea of trusted relays for
QKD~\cite{salvail:repeaters}. However, we prove that this is in fact all we can hope to achieve for secure two-party computation without any 
additional resources: Given only the resource of link-OT and classical communication, no protocol between Alice and Bob can be secure without the existence of an honest path.
Furthermore, we show that link-OT is in fact a necessary condition for any protocol to be secure, i.e., 
having access to a large network of nodes does not allow us to solve the problem of oblivious transfer on its own~\footnote{Note that
there could be only a single node on the honest path from Alice to Bob, which means we cannot simply solve this problem by assuming only very few
of the nodes are dishonest as in secure multi party computation.}.

Then, we successively relax the assumption of an honest path.
First, since relying on an honest path alone may still be rather unsatisfactory, we show that we can add a security backup in the sense that our protocol can be made
secure if there is either an honest path, or at least the dishonest node cannot break a computational assumption.
We then show that the assumption of the honest path can be relaxed if each pair of nodes are given a classical shared key for free, 
%In fact, this is exactly the underlying classical model for secure multi-party computation. Here, we enhance the classical model by adding the neighbouring quantum links and 
and finally that a non-trivial notion of security is still achievable for a node, even if everyone else in the network is dishonest. 

Our results open the door for extending implementations of oblivious transfer in the \nsw\ to large distances similar to the case of QKD~\cite{secoqc}.

% Steph: keep this for submission
%\smallskip
%\noindent
%{\bf The protocol:} 

\section{The protocol}

Let us first explain the problem of oblivious transfer (OT)~\cite{rabin:OT}; a formal definition can be found in~\cite{noisy:new}.
Alice (the sender) holds two input strings $s_0,s_1 \in \01^\ell$~\footnote{Chosen uniformly at random, unknown to Bob.} 
and Bob (the receiver) holds a choice bit $c \in \01$. If both nodes are honest, 
Bob should receive the input of his choosing, $s_c$, at the end of the protocol. 
If Bob is honest, then our goal is to ensure
that whatever attack Alice may mount, she can nevertheless not gain any information about $c$. Conversely, if Alice is honest,
we want that a dishonest Bob is unable to gain any information about at least one of Alice's inputs, $s_{1-c}$.
Whereas oblivious transfer by itself may seem like a rather obscure task, it has in fact been shown that Alice and Bob can use it to solve
any other cryptographic problem securely~\cite{kilian:foundingOnOT}. 
Below we use $\mbox{OT}((s_0,s_1),c)$ to indicate that we use a link-OT protocol as a black box.

We now provide two protocols, where the first is unconditionally secure for the sender Alice~\footnote{It is secure even if there is no honest path.}
and secure for the receiver Bob provided there is an honest path. The second has exactly opposite security properties: it is unconditionally secure for 
Bob and secure for Alice provided there is an honest path.
Let $N$ be the number of paths connecting Alice to Bob and denote by $v_1,\ldots,v_N$ the nodes adjacent to Alice on the $N$ possible paths.
We use '$+$' and '$\cdot$' to indicate bitwise addition and multiplication modulo $2$ respectively, and $\in_R$ to denote that a variable has been chosen uniformly and independently at random.

\begin{protocol}{}{Input: $s_0,s_1 \in \01^\ell$ for Alice, $c \in \01$ for Bob Output: $s_c$ to Bob.}
\label{proto:pathOT}
\item[{\bf Bob:}] Chooses $N$ bits ${c}_1,\ldots,{c}_{N} \in_R \01$ such that
$c = {c}_{1} + \ldots + {c}_N$. He sends ${c}_j$ to node $v_j$ along the $j$-th path.
\item[{\bf Alice:}]  Chooses $N$ keys $r_1,\ldots,r_{N} \in_R \01^\ell$ such that $r_1+\ldots +r_N=0$. 

	Performs $\mbox{OT}((t_{0,1},t_{1,1}),{c}_1)$ with $t_{0,1} = s_0 + r_1$ and $t_{1,1} = s_1 + r_1$ with node $v_1$,
	and $\mbox{OT}((t_{0,j},t_{1,j}),{c}_j)$ with
	$t_{0,j} = r_j$ and $t_{1,j} = s_0 + s_1 + r_j$ with nodes $v_2,\ldots,v_{N}$.
\item[{\bf Intermediary Nodes:}] Node $v_j$ sends $t_{{c}_j,j}$ to Bob along the $j$-th path.
\item[{\bf Bob:}] Computes $s_c = t_{{c}_1,1} + \ldots + t_{{c}_N,N}$.
\end{protocol}

First, the protocol is correct when both players are honest, since Bob computes
\begin{equation}
	s_c  =  t_{{c}_1,1} + \ldots + t_{{c}_N,N} = s_0+(s_0+s_1)\cdot c \label{eq:scDefine}
\end{equation}  

We now argue that the protocol is also secure; details can be found in the appendix.
Suppose first that Bob is dishonest. Since Alice uses fresh keys $\{r_j\}_j$ in each round, and $s_0$ and $s_1$ are themselves randomly chosen bit strings unknown to Bob, 
Bob would need to retrieve at least $N+1$ entries $t_{c_j,j}$ in order to compute both $s_0$ and $s_1$. However, even if Bob is working together with
all intermediary nodes, he can only learn at most one of $(t_{0,j},t_{1,j})$ from each of the $N$ link-OTs 
with Alice. Hence, Bob learns nothing about one of $s_{0}$ or $s_1$ as desired.

Suppose now that Bob is honest, and 
there exists an honest path between Alice and Bob. Note that Bob effectively performs a secret sharing of his input along all paths, so that Alice needs all shares in order to recover $c$~\cite{Sha79}. However, the share on the honest-path remains unknown to Alice. The security of the link-OT
ensures she cannot use it to gain any information about $c$ either.

Our second protocol is similar, but with Alice performing a secret sharing of her inputs.
Let $w_1,\ldots,w_N$ be the nodes adjacent to Bob on the $N$ possible paths.

\begin{protocol}{}{Input: $s_0,s_1 \in \01^\ell$ for Alice, $c \in \01$ for Bob Output: $s_c$ to Bob.}
\label{proto:longOT2}
\item[{\bf Alice:}] Chooses $N$ strings ${s}_{01},\ldots,{s}_{0N} \in_R \01^\ell$ such that $s_0 = {s}_{01} + \ldots + {s}_{0N}$ and 
	similarly ${s}_{11},\ldots,{s}_{1N} \in_R \01^\ell$ such that $s_1 = {s}_{11} + \ldots + {s}_{1N}$. She sends bits $s_{0j},s_{1j}$ to node $w_j$, i.e. the $j$-th neighbour of Bob via the $j$-th path. 
\item[{\bf Bob:}] Performs $\mbox{OT}( (s_{0j},s_{1j}),c)$ with node $w_j$.
	Computes $s_c = s_{c1} + \ldots + s_{cN}$.
\end{protocol}

Clearly, the protocol is correct if both parties are honest.
The security of the link-OT for the receiver ensures that even if a dishonest Alice controls all nodes adjacent to Bob, she nevertheless cannot learn $c$. 
Finally, the protocol is secure against a dishonest Bob, assuming that there exists an honest path: 
In this case, at least one of the shares $s_{0j}$ or $s_{1j}$ remains unknown to Bob, since they are securely transmitted to node $w_j$ via the honest path, and the link-OT protocol between $w_j$ and Bob is secure for the sender. Hence, he cannot learn both inputs $s_0,s_1$.

One may wonder whether we could have constructed a path-OT protocol without relying on the existence of a link-OT protocol, which is impossible
to obtain without assumptions~\cite{lo:insecurity}.
However, it is easy to see that the existence of \emph{any} path-OT protocol  would imply a secure link-OT protocol between two directly connected
parties, Anne and Bill:
First, Anne picks a path from Alice to Bob in the original setting. Then Bill picks a path from
Bob to Alice. 
The remaining paths they split arbitrarily. Now Anne acts
as Alice would and in addition simulates the action of all nodes in the paths assigned
to her. Bill also simulates the actions of Bob together with all nodes in the paths
assigned to him. Clearly, no matter who will be dishonest, we are
always in the setting where there is an honest path in the original
protocol, as one path is always simulated by someone being honest.
This means that we cannot hope to achieve OT in the honest-path model
without additional assumptions either~\cite{lo:insecurity}.

\section{Security without an honest path}

However, one might still hope that given such a strong primitive as link-OT we might be able to achieve security using only classical communication, even without the assumption of an honest path.
Unfortunately, it turns out that an honest path is indeed a necessary condition for security:
If there is no honest path, then there exists a subset of corrupted nodes $M$, such that any communication between Alice and Bob goes through them. 
Intuitively (see appendix for details) this means that either $M$ can gain information about $c$, or else must know enough about $s_0$ and $s_1$ 
to be able to supply Bob with the desired output. In the first case, dishonest Alice can learn $c$ from $M$, and in the second dishonest Bob
can break security by learning information about both of Alice's inputs.

{\em A security backup: } Nevertheless, the assumption of an honest path may appear quite strong, and it would be useful to 
have some security guarantees even if this assumption fails.
Fortunately, it is straightforward to adopt existing techniques from classical cryptography~\cite{HKNRR05, juerg:combiner} to extend
our protocols to be secure if either the honest-path assumption holds, or else if the dishonest party cannot break a certain computational 
problem. To this end, we combine our protocol with a protocol for classical oblivious transfer. OT can be achieved classically under a large 
variety of assumptions. Here, we choose to combine our protocol with the protocol of Naor and Pinkas~\cite{naor:efficientOT}, which is secure against
a dishonest sender if he cannot break the decisional Diffie Hellman problem (DDH), and unconditionally secure against a dishonest receiver.
Note that this means that just like for our honest-path assumption, we have unconditional security against one party, and security according to either
the DDH or the honest-path assumption against the other~\footnote{We can similarly construct a protocol that is secure against a dishonest receiver if there exists an honest path or he cannot break the decisional Diffie Hellman problem (DDH), and unconditionally secure against a dishonest sender.}.
Using the $\{3,2\}$-robust uniform OT-combiner from~\cite[Theorem 2]{juerg:combiner} we hence immediately obtain that
there exists an oblivious transfer protocol that is secure if either the honest path or the DDH assumption holds using
two instances of protocol~1, and two instances of the OT protocol of~\cite{naor:efficientOT}.
An explicit protocol can be found in~\cite{juerg:combiner}. 

{\em Secret keys: }
In the classical model for secure multiparty computation one usually assumes that there exist private links between all nodes 
and we are trying to show security against subsets of dishonest nodes. 
Clearly, this is a strong assumption as it requires 
us to establish keys over potentially long distances.
Nevertheless, it is interesting to consider a hybrid-model, 
where there exists a complete network of classical private links and also a network of quantum links between neighboring nodes allowing them to perform
link-OT. It is easy to see that our protocol can be transformed to achieve security as long as one of the neighbours of Alice and Bob is honest, instead of the entire
path being honest: we use the private channels to directly communicate with the immediate neighbours instead of relying on the entire path.
This easy example shows that allowing link-OT is indeed more powerful than what one can hope to gain in the classical model of secure multi-party computation.

{\em No assumptions: }
Finally, let us consider what happens if we allow an \emph{arbitrary} number of network nodes to be dishonest.
Curiously, some weak notion of security still remains. More specifically, by performing our two protocols sequentially with 
different inputs for Alice and Bob in the two executions, we can trivially construct a form of path-OT where Alice has four inputs, and Bob has two choice bits 
such that: if everyone is honest, then Bob learns two of the four bits, and Alice learns nothing about Bob's two index bits. If Alice is honest, but everyone else in the network is dishonest, then Bob learns three bits, but \emph{not} all four of them. If Bob is honest, but everyone else is dishonest, then Alice learns one of the two index bits of Bob, but \emph{not} both of them. These properties follow directly from our previous analysis. 

Note that this weak form of security is still impossible classically on a complete network with private links, unless computational assumptions are added. In our model, it becomes possible because we added the neighboring quantum links and assumed that we can perform short distance OT protocols via these quantum links. One 
can turn this weak OT protocol into some weak bit commitment protocol as well, leading to weak forms of coin tossing over long distances. 

\section{Conclusions}

We have shown security against dishonest Alice (or Bob) whenever there is at least one honest path, or the dishonest party cannot break a computational assumption.
One can easily extend our protocol to be robust against the case where the intermediary nodes may be dishonest independently of Alice
and Bob, and try to alter Alice's or Bob's input.
In our present protocols this is of course possible since they could for example flip one of the bits $\{c_j\}_j$.
To make the protocol robust we can simply use a more advanced secret sharing scheme that, similar to an error correcting code, 
protects against `errors' introduced in the secrets~\cite{CGMA85}. Note that depending on our choice
of secret sharing scheme, we may require more than one honest path to achieve robustness or more communciation rounds~\footnote{By letting Alice and Bob use some runs of the protocol to detect tampering by any intermediary nodes.}.

Our protocols show that two-party cryptographic primitives can be implemented over long distances in an extremely simple manner.
Our result enables us to extend the range of protocols in the \nsw\ in a similar way as has been done in QKD~\cite{secoqc}. 
Clearly, our protocols still require a considerable amount of classical communication. However, this is technologically much 
easier to achieve than entanglement swapping which of course still remains the more desirable solution. 
The quantum operations that the nodes are performing are no harder than the ones necessary in the link-OT protocols, i.e. it suffices that they create and measure BB84~\cite{bb84} states~\cite{noisy:new}. 
No complicated operations like Bell state measurements, or memory are required. 

\acknowledgments
SW was supported by NSF grants PHY-04056720, PHY-0803371, as well as 
the National Research Foundation and the Ministry of Education, Singapore. 
IK was supported by ANR grants ANR-09-JCJC-0067-01 and ANR-08-EMER-012.
Part of this work was done while SW was at the Institute for Quantum Information, Caltech.

\bigskip

\section{Security of protocol 1}
Whereas we ideally show security using the formal definitions for fully-randomized OT~\cite{noisy:new}, we restrict ourselves
to the simple arguments below in order to not obscure our argument, which is sufficient since
our setting is very straightforward to analyze.

\begin{claim}
Protocol 1 forms a secure oblivious transfer scheme with unconditional security against Alice, and security against Bob whenever
there exists an honest path.
\end{claim}

\begin{proof}
	We first show that the protocol is correct when both Alice and Bob are honest. This follows immediately by noting that Bob can compute

\begin{eqnarray}
        s_c & = & t_{{c}_1,1} + \ldots + t_{{c}_N,N} \label{eq:scDefine2}\\
& = & (s_0+(s_0+s_1)\cdot c_1 + r_1) + \sum_{i=2}^N((s_0+s_1)\cdot c_i+r_2)\nonumber\\
& = & s_0+(s_0+s_1)\cdot c \;\; = \;\; s_c\ .\nonumber
\end{eqnarray} 

We now show that the protocol is secure if Alice is honest, where 
we allow all intermediary players and Bob to be dishonest. 
From the security of the link-OT protocol, it follows that Bob can learn at most one of Alice's inputs from each invocation. 
In the most general cheating strategy, Bob can arbitrarily choose values as his input bits to the $N$ link-OT protocols. 
Let $d_1,d_2,\ldots,d_N$ denote these inputs and let $t_{{d}_1,1},\ldots,t_{{d}_N,N}$ be the inputs of Alice that Bob learns. 
Note that for any choice of Bob's inputs $\{d_i\}_i$ there exists a $c \in \01$ such that
$t_{{d}_1,1} + \ldots + t_{{d}_N,N} = s_c$. 
Moreover, $t_{{d}_1,1} + \ldots + t_{{d}_{N-1},N-1}+ t_{1-{d}_{N},N} = s_{1-{c}}$. 
Our goal is now to show that Bob cannot gain any information
about $s_{1-c}$.
First of all, note that since Alice uses fresh keys $\{r_j\}_j$ in each link-OT, and $s_0$ and $s_1$ are themselves randomly chosen bit strings 
unknown to Bob, the values of 
$t_{ {d}_1,1},\ldots,t_{ {d}_N,N}$ and $t_{1-{d}_N,N}$ are all independent.
Hence, Bob would need to retrieve all such $N+1$ entries in order to compute both $s_0$ and $s_1$, which contradicts the security of the link-OT.

It remains to prove security if Bob is honest. 
Note that Bob effectively performs a secret sharing of his input 
\begin{align}
	c = \sum_{j \in \{1,\ldots,N\}} c_j
\end{align}
along all paths such that
the bit $c$ can only be recovered if and only if Alice learns all shares $\{c_j\}_j$. However, Alice has no information about the value of $c_j$ on the honest-path as links between honest players are secure.
Furthermore, the link-OT used between Alice and $v_j$ is secure for the receiver, and hence
we conclude that Alice cannot learn $c$ as promised.
\end{proof}

\section{Necessity of the honest path}
We now prove that an honest-path is a necessary condition for OT, where we use a weaker definition which is implied by the formal ones given e.g. 
for (fully randomized) oblivious transfer in~\cite{noisy:new}. Note that this is sufficient to prove the impossibility of the more difficult task
as well. 
More concretely, the following conditions must hold for any protocol that is both correct and secure.
Any impossibility proof for a protocol aiming for perfect security is rather unsatisfactory since we would be willing to accept
a very small probability of failure. We hence include a security parameter $\eps > 0$ which intuitively corresponds to the error 
we are willing to accept.

First of all, for any protocol that is correct we must have
that the probability that honest Bob with input $c$ can guess honest Alice's input, $s_{c}$, satisfies
\begin{align}
	\mbox{Correctness: } \Pr[s_{c}|Bob] \geq 1 - \eps\ .
	\label{eq:correct}
\end{align}
Furthermore, if Alice is honest, then for whatever attack Bob may conceive we have that he cannot guess at least one of the two inputs
\begin{align}
	\mbox{Security against Bob: } \exists b\ \Pr[s_{b}|Bob] \leq \frac{1}{2^\ell} + \eps\ ,
	\label{eq:Bobcheats}
\end{align}
Finally, if Bob is honest and his input bit is $c$, then for any strategy of dishonest Alice, she is unable to learn Bob's choice bit
\begin{align}
	\mbox{Security against Alice: }  \Pr[{c}|Alice] \leq \frac{1}{2} + \eps\ .
	\label{eq:Alicecheats}
\end{align}

	To obtain an impossibility proof, our goal is now to 
	show that~\eqref{eq:correct},~\eqref{eq:Bobcheats} and~\eqref{eq:Alicecheats} can never be satisfied simultaneously 
	for small values of $\eps$. That is, we can only hope to achieve very imperfect version of oblivious transfer with a large error $\eps$.

\begin{claim}
There exists no protocol for oblivious transfer based on only link-OT and classical communication that is secure without an honest path between Alice and Bob
with security parameter $\eps < 1/4 - 1/2^{\ell+2}$. 
\end{claim}

\begin{proof}
	
	If there is no honest path, then there exists some subset of potentially dishonest nodes $M$ that separates the
	network into two disconnected components, one containing Alice and the other Bob. 
Let us now establish some basic properties of the probabilities that Alice, Bob or $M$ can learn $s_0,s_1$ and $c$ in an honest execution of any protocol. 

Note that in any protocol, Bob cannot gain more information about Alice's inputs than $M$ can, since all information between Alice and Bob runs through $M$ (wlog we can furthermore assume that dishonest Bob would give any shared secret keys with Alice to $M$ for free). 
Hence, we have that
	\begin{align}
		\forall b\ \Pr[s_{b}|Bob] \leq \Pr[s_{b}|M]\ .
		\label{eq:conditionA}
	\end{align}
Similarly, Alice cannot gain more information about Bob's input than $M$ can, hence
\begin{align}
		\Pr[{c}|Alice] \leq \Pr[{c}|M]\ .
		\label{eq:conditionB}
	\end{align}
First, suppose that for an honest execution of any protocol the probability that $M$ is able to guess $c$ satisfies $\Pr[c|M] > 1/2 + \eps$. Then, Alice can violate the security condition~\eqref{eq:Alicecheats} by running the protocol honestly with Bob and then asking $M$ for a guess of $c$. Hence, it must hold that 
	\begin{align}
		\Pr[c|M] \leq \frac{1}{2} + \eps\ .
		\label{eq:conditionC}
	\end{align}
Second, by the correctness condition~\eqref{eq:correct} and equation~\eqref{eq:conditionA}, for an honest execution of any protocol, we have
\begin{align}
	\Pr[s_{c}|M] \geq 1 - \eps\ .
	\label{eq:conditionD}
\end{align}
Third, suppose that for an honest execution of any protocol, $\Pr[s_{1-c}|M] > \frac{1}{2^\ell} + \eps$. Then, Bob can violate the security condition~\eqref{eq:Bobcheats} by running the protocol honestly with Alice and then asking $M$ for a guess for both $s_0$ and $s_1$. Hence, it must hold that
\begin{align}
	\Pr[s_{1-c}|M] \leq \frac{1}{2^\ell} + \eps\ .
	\label{eq:conditionE}
\end{align}
	We now show that these conditions imply that whenever Bob is honest, there exists a cheating strategy for Alice.
	Alice first chooses two random inputs $s_0,s_1 \in \01^\ell$, and runs the protocol as an honest Alice would do. 
	Afterwards, she picks a random $b$ and asks $M$, who by definition will willingly cooperate with any cheating party, to send her a guess 
	$\tilde{s}_b$ for $s_b$. Note that~\eqref{eq:conditionD} and~\eqref{eq:conditionE} now tell us
	that the probability that $M$ succeeds is very large for $s_{c}$, but 
	extremely small for $s_{1-c}$. Alice then outputs $b$ as her guess for $c$ if $M$ guessed correctly and $1-b$ if $M$ guessed wrongly.
	The probability that Alice succeeds using this strategy obeys
	\begin{align}
		\lefteqn{\Pr[c|Alice]}\\
		 &= \Pr[b=c]\Pr[s_{c}|M]+\Pr[b=1-c] (1 - \Pr[s_{1-c}|M])\nonumber\\
		&\geq \frac{1}{2}(1-\eps)+\frac{1}{2} (1 - \frac{1}{2^{\ell}}-\eps)\\
			      &= (1-\eps- \frac{1}{2^{\ell+1}} )
\label{eq:lowerAlice}\ .
	\end{align}
	Comparing~\eqref{eq:lowerAlice} with~\eqref{eq:Alicecheats} concludes our claim.
\end{proof}

Note, however, that OT \emph{is} of course possible if $M$ would be fully quantum, and in particular would be able to perform
entanglement swapping between Alice and Bob.

\begin{thebibliography}{26}
\expandafter\ifx\csname natexlab\endcsname\relax\def\natexlab#1{#1}\fi
\expandafter\ifx\csname bibnamefont\endcsname\relax
  \def\bibnamefont#1{#1}\fi
\expandafter\ifx\csname bibfnamefont\endcsname\relax
  \def\bibfnamefont#1{#1}\fi
\expandafter\ifx\csname citenamefont\endcsname\relax
  \def\citenamefont#1{#1}\fi
\expandafter\ifx\csname url\endcsname\relax
  \def\url#1{\texttt{#1}}\fi
\expandafter\ifx\csname urlprefix\endcsname\relax\def\urlprefix{URL }\fi
\providecommand{\bibinfo}[2]{#2}
\providecommand{\eprint}[2][]{\url{#2}}

\bibitem[{\citenamefont{Bennett and Brassard}(1984)}]{bb84}
\bibinfo{author}{\bibfnamefont{C.~H.} \bibnamefont{Bennett}} \bibnamefont{and}
  \bibinfo{author}{\bibfnamefont{G.}~\bibnamefont{Brassard}}, in
  \emph{\bibinfo{booktitle}{Proceedings of the IEEE International Conference on
  Computers, Systems and Signal Processing}} (\bibinfo{year}{1984}), pp.
  \bibinfo{pages}{175--179}.

\bibitem[{\citenamefont{Ekert}(1991)}]{e91}
\bibinfo{author}{\bibfnamefont{A.}~\bibnamefont{Ekert}},
  \bibinfo{journal}{Physical Review Letters} \textbf{\bibinfo{volume}{67}},
  \bibinfo{pages}{661} (\bibinfo{year}{1991}).

\bibitem[{\citenamefont{Wehner et~al.}(2008)\citenamefont{Wehner, Schaffner,
  and Terhal}}]{prl:noisy}
\bibinfo{author}{\bibfnamefont{S.}~\bibnamefont{Wehner}},
  \bibinfo{author}{\bibfnamefont{C.}~\bibnamefont{Schaffner}},
  \bibnamefont{and} \bibinfo{author}{\bibfnamefont{B.~M.}
  \bibnamefont{Terhal}}, \bibinfo{journal}{Physical Review Letters}
  \textbf{\bibinfo{volume}{100}}, \bibinfo{pages}{220502}
  (\bibinfo{year}{2008}).

\bibitem[{\citenamefont{K{\"o}nig et~al.}(2009)\citenamefont{K{\"o}nig, Wehner,
  and Wullschleger}}]{noisy:new}
\bibinfo{author}{\bibfnamefont{R.}~\bibnamefont{K{\"o}nig}},
  \bibinfo{author}{\bibfnamefont{S.}~\bibnamefont{Wehner}}, \bibnamefont{and}
  \bibinfo{author}{\bibfnamefont{J.}~\bibnamefont{Wullschleger}}
  (\bibinfo{year}{2009}), \bibinfo{note}{arXiv:0906.1030}.

\bibitem[{\citenamefont{Schaffner}(2010)}]{chris:new}
\bibinfo{author}{\bibfnamefont{C.}~\bibnamefont{Schaffner}}
  (\bibinfo{year}{2010}), \bibinfo{note}{arXiv:1002.1495. Workshop on
  Cryptography from Storage Imperfections, Recent work, Caltech, March 20--22}.

\bibitem[{\citenamefont{Damg{\aa}rd et~al.}(2007)\citenamefont{Damg{\aa}rd,
  Fehr, Salvail, and Schaffner}}]{bounded:secureId}
\bibinfo{author}{\bibfnamefont{I.~B.} \bibnamefont{Damg{\aa}rd}},
  \bibinfo{author}{\bibfnamefont{S.}~\bibnamefont{Fehr}},
  \bibinfo{author}{\bibfnamefont{L.}~\bibnamefont{Salvail}}, \bibnamefont{and}
  \bibinfo{author}{\bibfnamefont{C.}~\bibnamefont{Schaffner}}, in
  \emph{\bibinfo{booktitle}{Advances in Cryptology---CRYPTO~'07}}
  (\bibinfo{publisher}{Springer-Verlag}, \bibinfo{year}{2007}), vol.
  \bibinfo{volume}{4622} of \emph{\bibinfo{series}{Lecture Notes in Computer
  Science}}, pp. \bibinfo{pages}{342--359}.

\bibitem[{\citenamefont{Yao}(1982)}]{yao:sfe}
\bibinfo{author}{\bibfnamefont{A.~C.} \bibnamefont{Yao}}, in
  \emph{\bibinfo{booktitle}{Proceedings of the 23rd Annual IEEE FOCS}}
  (\bibinfo{year}{1982}), pp. \bibinfo{pages}{160--164}.

\bibitem[{\citenamefont{Rosenberg et~al.}(2008)\citenamefont{Rosenberg,
  Rosenberg, Harrington, Rice, Dallman, Tyagi, McCabe, Hughes, Nordholt,
  Hadfield et~al.}}]{hughes:fiber}
\bibinfo{author}{\bibfnamefont{D.}~\bibnamefont{Rosenberg}},
  \bibinfo{author}{\bibfnamefont{C.}~\bibnamefont{Rosenberg}},
  \bibinfo{author}{\bibfnamefont{J.}~\bibnamefont{Harrington}},
  \bibinfo{author}{\bibfnamefont{P.}~\bibnamefont{Rice}},
  \bibinfo{author}{\bibfnamefont{N.}~\bibnamefont{Dallman}},
  \bibinfo{author}{\bibfnamefont{K.}~\bibnamefont{Tyagi}},
  \bibinfo{author}{\bibfnamefont{K.}~\bibnamefont{McCabe}},
  \bibinfo{author}{\bibfnamefont{R.}~\bibnamefont{Hughes}},
  \bibinfo{author}{\bibfnamefont{J.}~\bibnamefont{Nordholt}},
  \bibinfo{author}{\bibfnamefont{R.}~\bibnamefont{Hadfield}},
  \bibnamefont{et~al.}, in \emph{\bibinfo{booktitle}{Proceedings of the
  Conference on Optical Fiber communication}} (\bibinfo{year}{2008}), pp.
  \bibinfo{pages}{1--3}.

\bibitem[{\citenamefont{Schmitt-Manderbach
  et~al.}(2007)\citenamefont{Schmitt-Manderbach, Weier, FÃ¼rst, Ursin,
  Tiefenbacher, Scheidl, Perdigues, Sodnik, Rarity, Zeilinger
  et~al.}}]{qkd:freespace1}
\bibinfo{author}{\bibfnamefont{T.}~\bibnamefont{Schmitt-Manderbach}},
  \bibinfo{author}{\bibfnamefont{H.}~\bibnamefont{Weier}},
  \bibinfo{author}{\bibfnamefont{M.}~\bibnamefont{FÃ¼rst}},
  \bibinfo{author}{\bibfnamefont{R.}~\bibnamefont{Ursin}},
  \bibinfo{author}{\bibfnamefont{F.}~\bibnamefont{Tiefenbacher}},
  \bibinfo{author}{\bibfnamefont{T.}~\bibnamefont{Scheidl}},
  \bibinfo{author}{\bibfnamefont{J.}~\bibnamefont{Perdigues}},
  \bibinfo{author}{\bibfnamefont{Z.}~\bibnamefont{Sodnik}},
  \bibinfo{author}{\bibfnamefont{J.~G.} \bibnamefont{Rarity}},
  \bibinfo{author}{\bibfnamefont{A.}~\bibnamefont{Zeilinger}},
  \bibnamefont{et~al.}, \bibinfo{journal}{Physical Review Letters}
  \textbf{\bibinfo{volume}{98}}, \bibinfo{pages}{010504}
  (\bibinfo{year}{2007}).

\bibitem[{\citenamefont{Ursin et~al.}(2007)\citenamefont{Ursin, Tiefenbacher,
  Schmitt-Manderbach, Weier, Scheidl, Lindenthal, Blauensteiner, Jennewein,
  Perdigues, Trojek et~al.}}]{zeilinger:qkd}
\bibinfo{author}{\bibfnamefont{R.}~\bibnamefont{Ursin}},
  \bibinfo{author}{\bibfnamefont{F.}~\bibnamefont{Tiefenbacher}},
  \bibinfo{author}{\bibfnamefont{T.}~\bibnamefont{Schmitt-Manderbach}},
  \bibinfo{author}{\bibfnamefont{H.}~\bibnamefont{Weier}},
  \bibinfo{author}{\bibfnamefont{T.}~\bibnamefont{Scheidl}},
  \bibinfo{author}{\bibfnamefont{M.}~\bibnamefont{Lindenthal}},
  \bibinfo{author}{\bibfnamefont{B.}~\bibnamefont{Blauensteiner}},
  \bibinfo{author}{\bibfnamefont{T.}~\bibnamefont{Jennewein}},
  \bibinfo{author}{\bibfnamefont{J.}~\bibnamefont{Perdigues}},
  \bibinfo{author}{\bibfnamefont{P.}~\bibnamefont{Trojek}},
  \bibnamefont{et~al.}, \bibinfo{journal}{Nature Physics}
  \textbf{\bibinfo{volume}{3}}, \bibinfo{pages}{481} (\bibinfo{year}{2007}).

\bibitem[{\citenamefont{Lloyd et~al.}(2004)\citenamefont{Lloyd, Shapiro, Wong,
  Kumar, Shahriar, and Yuen}}]{lloyd:internet}
\bibinfo{author}{\bibfnamefont{S.}~\bibnamefont{Lloyd}},
  \bibinfo{author}{\bibfnamefont{J.~H.} \bibnamefont{Shapiro}},
  \bibinfo{author}{\bibfnamefont{F.~N.~C.} \bibnamefont{Wong}},
  \bibinfo{author}{\bibfnamefont{P.}~\bibnamefont{Kumar}},
  \bibinfo{author}{\bibfnamefont{S.~M.} \bibnamefont{Shahriar}},
  \bibnamefont{and} \bibinfo{author}{\bibfnamefont{H.}~\bibnamefont{Yuen}},
  \bibinfo{journal}{ACM SIGCOMM Computer Communication Review}
  \textbf{\bibinfo{volume}{34}}, \bibinfo{pages}{9} (\bibinfo{year}{2004}).

\bibitem[{\citenamefont{Kimble}(2008)}]{kimble:internet}
\bibinfo{author}{\bibfnamefont{H.~J.} \bibnamefont{Kimble}},
  \bibinfo{journal}{Nature} \textbf{\bibinfo{volume}{453}},
  \bibinfo{pages}{1023} (\bibinfo{year}{2008}).

\bibitem[{\citenamefont{Alleaume et~al.}(2007)\citenamefont{Alleaume, Bouda,
  Branciard, Debuisschert, Dianati, Gisin, Godfrey, Grangier, Langer, Leverrier
  et~al.}}]{secoqc}
\bibinfo{author}{\bibfnamefont{R.}~\bibnamefont{Alleaume}},
  \bibinfo{author}{\bibfnamefont{J.}~\bibnamefont{Bouda}},
  \bibinfo{author}{\bibfnamefont{C.}~\bibnamefont{Branciard}},
  \bibinfo{author}{\bibfnamefont{T.}~\bibnamefont{Debuisschert}},
  \bibinfo{author}{\bibfnamefont{M.}~\bibnamefont{Dianati}},
  \bibinfo{author}{\bibfnamefont{N.}~\bibnamefont{Gisin}},
  \bibinfo{author}{\bibfnamefont{M.}~\bibnamefont{Godfrey}},
  \bibinfo{author}{\bibfnamefont{P.}~\bibnamefont{Grangier}},
  \bibinfo{author}{\bibfnamefont{T.}~\bibnamefont{Langer}},
  \bibinfo{author}{\bibfnamefont{A.}~\bibnamefont{Leverrier}},
  \bibnamefont{et~al.} (\bibinfo{year}{2007}),
  \bibinfo{note}{arXiv:quant-ph/0701168}.

\bibitem[{\citenamefont{\ifmmode~\dot{Z}\else \.{Z}\fi{}ukowski
  et~al.}(1993)\citenamefont{\ifmmode~\dot{Z}\else \.{Z}\fi{}ukowski,
  Zeilinger, Horne, and Ekert}}]{entSwapping}
\bibinfo{author}{\bibfnamefont{M.}~\bibnamefont{\ifmmode~\dot{Z}\else
  \.{Z}\fi{}ukowski}},
  \bibinfo{author}{\bibfnamefont{A.}~\bibnamefont{Zeilinger}},
  \bibinfo{author}{\bibfnamefont{M.~A.} \bibnamefont{Horne}}, \bibnamefont{and}
  \bibinfo{author}{\bibfnamefont{A.~K.} \bibnamefont{Ekert}},
  \bibinfo{journal}{Physical Review Letters} \textbf{\bibinfo{volume}{71}},
  \bibinfo{pages}{4287} (\bibinfo{year}{1993}).

\bibitem[{\citenamefont{Pan et~al.}(1998)\citenamefont{Pan, Bouwmeester,
  Weinfurter, and Zeilinger}}]{zeilinger:entanglement}
\bibinfo{author}{\bibfnamefont{J.}~\bibnamefont{Pan}},
  \bibinfo{author}{\bibfnamefont{D.}~\bibnamefont{Bouwmeester}},
  \bibinfo{author}{\bibfnamefont{H.}~\bibnamefont{Weinfurter}},
  \bibnamefont{and}
  \bibinfo{author}{\bibfnamefont{A.}~\bibnamefont{Zeilinger}},
  \bibinfo{journal}{Physical Review Letters} \textbf{\bibinfo{volume}{80}},
  \bibinfo{pages}{3891} (\bibinfo{year}{1998}).

\bibitem[{\citenamefont{Riebe et~al.}(2008)\citenamefont{Riebe, Monz, Kim,
  Villar, Schindler, Chwalla, Hennrich, and Blatt}}]{entSwapex2}
\bibinfo{author}{\bibfnamefont{M.}~\bibnamefont{Riebe}},
  \bibinfo{author}{\bibfnamefont{T.}~\bibnamefont{Monz}},
  \bibinfo{author}{\bibfnamefont{K.}~\bibnamefont{Kim}},
  \bibinfo{author}{\bibfnamefont{A.~S.} \bibnamefont{Villar}},
  \bibinfo{author}{\bibfnamefont{P.}~\bibnamefont{Schindler}},
  \bibinfo{author}{\bibfnamefont{M.}~\bibnamefont{Chwalla}},
  \bibinfo{author}{\bibfnamefont{M.}~\bibnamefont{Hennrich}}, \bibnamefont{and}
  \bibinfo{author}{\bibfnamefont{R.}~\bibnamefont{Blatt}},
  \bibinfo{journal}{Nature Physics} \textbf{\bibinfo{volume}{4}},
  \bibinfo{pages}{839} (\bibinfo{year}{2008}).

\bibitem[{\citenamefont{Goebel et~al.}(2008)\citenamefont{Goebel, Wagenknecht,
  Zhang, Chen, Chen, Schmiedmayer, and Pan}}]{entSwapex3}
\bibinfo{author}{\bibfnamefont{A.~M.} \bibnamefont{Goebel}},
  \bibinfo{author}{\bibfnamefont{C.}~\bibnamefont{Wagenknecht}},
  \bibinfo{author}{\bibfnamefont{Q.}~\bibnamefont{Zhang}},
  \bibinfo{author}{\bibfnamefont{Y.}~\bibnamefont{Chen}},
  \bibinfo{author}{\bibfnamefont{K.}~\bibnamefont{Chen}},
  \bibinfo{author}{\bibfnamefont{J.}~\bibnamefont{Schmiedmayer}},
  \bibnamefont{and} \bibinfo{author}{\bibfnamefont{J.}~\bibnamefont{Pan}},
  \bibinfo{journal}{Physical Review Letters} \textbf{\bibinfo{volume}{101}},
  \bibinfo{pages}{080403} (\bibinfo{year}{2008}).

\bibitem[{\citenamefont{Kilian}(1988)}]{kilian:foundingOnOT}
\bibinfo{author}{\bibfnamefont{J.}~\bibnamefont{Kilian}}, in
  \emph{\bibinfo{booktitle}{Proceedings of 20th ACM STOC}}
  (\bibinfo{year}{1988}), pp. \bibinfo{pages}{20--31}.

\bibitem[{\citenamefont{Salvail et~al.}(2010)\citenamefont{Salvail, Peev,
  Diamanti, Alleaume, L{\"u}tkenhaus, and L{\"a}nger}}]{salvail:repeaters}
\bibinfo{author}{\bibfnamefont{L.}~\bibnamefont{Salvail}},
  \bibinfo{author}{\bibfnamefont{M.}~\bibnamefont{Peev}},
  \bibinfo{author}{\bibfnamefont{E.}~\bibnamefont{Diamanti}},
  \bibinfo{author}{\bibfnamefont{R.}~\bibnamefont{Alleaume}},
  \bibinfo{author}{\bibfnamefont{N.}~\bibnamefont{L{\"u}tkenhaus}},
  \bibnamefont{and}
  \bibinfo{author}{\bibfnamefont{T.}~\bibnamefont{L{\"a}nger}},
  \bibinfo{journal}{Journal of Computer Security}
  \textbf{\bibinfo{volume}{18}}, \bibinfo{pages}{61} (\bibinfo{year}{2010}).

\bibitem[{\citenamefont{Rabin}(1981)}]{rabin:OT}
\bibinfo{author}{\bibfnamefont{M.}~\bibnamefont{Rabin}}, \bibinfo{type}{Tech.
  Rep.}, \bibinfo{institution}{Aiken Computer Laboratory, Harvard University}
  (\bibinfo{year}{1981}), \bibinfo{note}{technical Report TR-81}.

\bibitem[{\citenamefont{Shamir}(1979)}]{Sha79}
\bibinfo{author}{\bibfnamefont{A.}~\bibnamefont{Shamir}},
  \bibinfo{journal}{Communications of the ACM} \textbf{\bibinfo{volume}{22}},
  \bibinfo{pages}{612} (\bibinfo{year}{1979}).

\bibitem[{\citenamefont{Lo}(1997)}]{lo:insecurity}
\bibinfo{author}{\bibfnamefont{H.-K.} \bibnamefont{Lo}},
  \bibinfo{journal}{Physical Review A} \textbf{\bibinfo{volume}{56}},
  \bibinfo{pages}{1154} (\bibinfo{year}{1997}).

\bibitem[{\citenamefont{Harnik et~al.}(2005)\citenamefont{Harnik, Kilian, Naor,
  Reingold, and Rosen}}]{HKNRR05}
\bibinfo{author}{\bibfnamefont{D.}~\bibnamefont{Harnik}},
  \bibinfo{author}{\bibfnamefont{J.}~\bibnamefont{Kilian}},
  \bibinfo{author}{\bibfnamefont{M.}~\bibnamefont{Naor}},
  \bibinfo{author}{\bibfnamefont{O.}~\bibnamefont{Reingold}}, \bibnamefont{and}
  \bibinfo{author}{\bibfnamefont{A.}~\bibnamefont{Rosen}}, in
  \emph{\bibinfo{booktitle}{Advances in Cryptology --- EUROCRYPT}}
  (\bibinfo{year}{2005}), vol. \bibinfo{volume}{3494} of
  \emph{\bibinfo{series}{Lecture Notes in Computer Science}}, pp.
  \bibinfo{pages}{96--113}.

\bibitem[{\citenamefont{Meier et~al.}(2007)\citenamefont{Meier, Przydatek, and
  Wullschleger}}]{juerg:combiner}
\bibinfo{author}{\bibfnamefont{R.}~\bibnamefont{Meier}},
  \bibinfo{author}{\bibfnamefont{B.}~\bibnamefont{Przydatek}},
  \bibnamefont{and}
  \bibinfo{author}{\bibfnamefont{J.}~\bibnamefont{Wullschleger}}, in
  \emph{\bibinfo{booktitle}{Theory of Cryptography Conference --- TCC}}
  (\bibinfo{year}{2007}), Lecture Notes in Computer Science.

\bibitem[{\citenamefont{Naor and Pinkas}(2001)}]{naor:efficientOT}
\bibinfo{author}{\bibfnamefont{M.}~\bibnamefont{Naor}} \bibnamefont{and}
  \bibinfo{author}{\bibfnamefont{B.}~\bibnamefont{Pinkas}}, in
  \emph{\bibinfo{booktitle}{Proceedings of 12th SODA}} (\bibinfo{year}{2001}),
  pp. \bibinfo{pages}{448--457}.

\bibitem[{\citenamefont{B.~Chor and Awerbuch}(1985)}]{CGMA85}
\bibinfo{author}{\bibfnamefont{S.~M.} \bibnamefont{B.~Chor},
  \bibfnamefont{S.~Goldwasser}} \bibnamefont{and}
  \bibinfo{author}{\bibfnamefont{B.}~\bibnamefont{Awerbuch}}, in
  \emph{\bibinfo{booktitle}{Proceedings of 44th IEEE FOCS}}
  (\bibinfo{year}{1985}), pp. \bibinfo{pages}{383--395}.

\end{thebibliography}
\end{document}